\documentclass[runningheads]{llncs}
\usepackage[T1]{fontenc}
\usepackage{cite}
\usepackage{amsmath,amsthm,amssymb,amsfonts}
\usepackage{amsmath}  
\usepackage{float}
\usepackage[ruled,vlined,boxed]{algorithm2e}

\usepackage{graphicx}
\usepackage{subfigure}
\usepackage{makecell}
\usepackage{enumerate}
\usepackage{multirow}
\usepackage{color}
\usepackage{slashed}
\usepackage{bbding}
\usepackage{booktabs}
\usepackage{subfigure}
\usepackage{pifont}

\usepackage[marginal]{footmisc}
\begin{document}


\title{NetChain: Authe\textcolor{black}{n}ticat\textcolor{black}{e}d Blockchain \textcolor{black}{T}op-$k$ Graph Data Queries and its Application in Asset Management}

\author{Hongguang Zhao\and
Xu Yang\and
Saiyu Qi\inst{(}\Envelope\inst{)} \and
Qiuhao Wang \and Ke Li}

%
\institute{Xi’an Jiaotong University, Xi’an, Shaanxi, China \\
\email{saiyu-qi@xjtu.edu.cn}}

\maketitle

\begin{abstract}
As a valuable digital resource, graph data is an important data asset, which has been widely utilized across various fields to optimize decision-making and enable smarter solutions. To manage data assets, blockchain is widely used to enable data sharing and trading, but it cannot supply complex analytical queries. vChain was proposed to achieve verifiable boolean queries over blockchain by designing an embedded authenticated data structure (ADS). However, for generating (non-)existence proofs, vChain suffers from expensive storage and computation costs in ADS construction, along with high communication and verification costs. In this paper, we propose a novel NetChain framework that enables efficient top-k queries over on-chain graph data with verifiability. Specifically, we design a novel authenticated two-layer index that supports (non-)existence proof generation in block-level and built-in verifiability for matched objects. To further alleviate the computation and verification overhead, an optimized variant NetChain$^+$ is derived. 
The authenticity of our frameworks is validated through security analysis.
Evaluations show that NetChain and NetChain$^+$ outperform vChain, respectively achieving up to $85\times$ and $31\times$ improvements on ADS construction.
Moreover, compared with vChain, NetChain$^+$ reduces the communication and verification costs by $87\%$ and $96\%$ respectively.
\end{abstract}

\setlength{\parskip}{0pt}

\section{Introduction}\label{introduction}

In the digital era, data assets have become essential resources, encompassing various forms such as IoT data, medical information, and many other types that hold immense potential for fostering innovation and developing new services \cite{data_asset}.
Among these, graph data stands out as a powerful type of data asset, which is widely applied in fields like social networks \cite{li2023rumor} to facilitate the analysis of user interactions and knowledge graphs \cite{yang2023knowledge} to promote applications like recommendation systems and intelligent search engines. Generally speaking,
Graph data’s ability to uncover complex relationships and power smarter, data-driven solutions underscores its significant contribution to advancing human development and innovation across industries.

In order to enlarge the value of graph data assets, different institutes tend to share their owned graph data\cite{shared}.
For example, several online communities could share social graph data to assist in building more precise recommendation systems.
Blockchain \cite{DBLP:journals/tkde/GuoXWWWJ24} is a fundamental technique enabling multiple participants in different trust domains to share and monetize data asset with security guarantee.
As a distributed ledger, blockchain can not only support data asset ownership confirmation, but also supply data storage and query service with integrity assurance \cite{consortium}.
However, mainstream
 blockchains only support verifiability for simple query types, without the ability to enable verifiable analytical queries over on-chain graph data, e.g., classic top-$k$ social network query.

Recently, many works have studied authenticated queries over blockchain stored data \cite{conference:vchain,liu2023veffchain,cheng2024lightweight,xu2023empowering} to provide correct blockchain-aided data query services.
As a representative work, vChain\cite{conference:vchain} implements verifiable boolean queries over blockchain.
In specific, vChain constructs a built-in merkle hash tree (MHT) on intra-block objects as the authenticated data structure (ADS), where each MHT node stores the union of keywords (denoted by $W$) in child nodes as well as the multiset's accumulative value of $W$ (denoted by $AttDigest$) to generate (non-)existence proofs as needed.
However, vChain suffers two efficiency bottlenecks when orienting towards on-chain graph data.
First, vChain suffers from expensive per-block ADS construction cost for generating non-existence proofs. To construct the ADS for a block, each MHT node needs to store an additional keywords set $W$ and the computation of $AttDigest$ involves complex asymmetric cryptography generation. Second, vChain incurs significant query authentication cost for generating existence proofs. Considering that a block contains multiple matched objects, vChain needs to generate merkle proofs for each of them as their existence proofs, incurring high communication and verification costs for the query.

To address these issues, in this paper, we propose an \emph{authenticated \mbox{top-$k$} graph query framework (NetChain)} over blockchain to provide efficient and authenticated query services over on-chain graph data. 
Comparing with vChain, We devise new ADS to achieve lightweight ADS construction cost (only involves hash functions) and lightweight query authentication cost (generates (non-)existence proofs in block level).
For a top-$k$ graph query, we divide blocks of the blockchain into (non-)matched blocks where a (non-)matched block (not) contains matched objects. We then propose a novel intra-block embedded ADS named \emph{authenticated two-layer index} to reduce ADS construction costs and deal with (non-)existence proofs efficiently.
To prove a (non-)matched block, we construct a \emph{compound key-based sorted merkle tree (SMT)} (the upper layer of the index) solely using hash functions on all intra-block compound keys (a novel concept to index clusters of matched objects), which not only reduces the ADS generation and storage costs but also enables efficient (non-)existence proof generation for the block.
Then for each matched block, we construct an \emph{ordered hash chain} (the bottom layer of the index) to organize all objects indexed by the same compound key in the block with built-in verifiability, avoiding the necessity to generate existence proofs for individual objects.  
Furthermore, we extend NetChain to an optimized variant NetChain$^{+}$ by customizing an \emph{inter-block link construction} and a \emph{two-round scan mechanism} to avoid the authentication for non-matched blocks and invalid objects, thus greatly reducing authentication cost. To summarize, the contributions in this paper are as follows:
\begin{itemize}
    \item[$\bullet$] We propose a novel framework termed NetChain to support authenticated top-$k$ graph queries over blockchain. To the best of our knowledge, this is the first attempt at graph-oriented blockchain verifiable query processing.
    \item[$\bullet$] We design a novel authenticated two-layer index, composed of a compound key-based SMT and a series of ordered hash chains, as the built-in ADS to support verifiable top-$k$ graph query.
    \item[$\bullet$] To further alleviate the authentication cost, we propose NetChain$^{+}$, which includes an inter-block link construction and a two-round scan mechanism.
    \item[$\bullet$] Finally, evaluation results exhibit that NetChain and NetChain$^+$ respectively yields up to $85 \times$ and $31\times$ improvements in ADS construction time  compared with the state-of-the-art framework vChain, and NetChain reduces the ADS size by $65\%$. 
    Meanwhile, compared with vChain, NetChain$^+$ reduces the communication and verification cost by $87\%$ and $96\%$ respectively.
    
\end{itemize}



\section{Related Work}
\noindent \textbf{Verifiable Query Over Blockchain.}
A series of solutions are proposed to enable verifiable queries over blockchain. 
Xu et al.\cite{conference:vchain} proposed vChain, the first work that uses embedded ADSs to achieve boolean range queries over blockchain with data integrity guarantee.
However, vChain utilizes expensive multiset accumulator to construct the ADS, causing significant computation and storage overhead.
Liu et al. \cite{liu2023veffchain} proposed \emph{veffChain} to supply latest-$k$ query over blockchain with freshness guarantee, but the parameter $k$ is pre-assigned and cannot be modified as needed.
Cheng et al.\cite{cheng2024lightweight} proposed the top-$k$ query processing on numeric attributes with specific keyword conditions for the first time, where the miner builds a sorted merkle tree for each keyword as the index, leading to heavy burden for full nodes.
Moreover, Xu et al.\cite{xu2023empowering} proposed an authenticated way to maintain the Spatial-Temporal-Keywords transactions in blockchain.
Zhang et al.\cite{zhang2024cole} designed a column-based learned storage for ETH, effectively facilitating the provenance queries for states.
Pang et al.\cite{pang2020authqx} supplied lightweight verifiable queries using trust execution environment to build hierarchical ADS.
However, none of the existing solutions focus on the query for on-chain graph data.

\noindent \textbf{Query On Outsourced Graph Data.}
Graph data is often outsourced to external providers due to data sharing or storage limitations, inducing the demand for querying on outsourced graph data.
Wu et al.\cite{wu2023enabling} proposed PAGB, a method that enables verifiable query for property graph data stored in hybrid-storage blockchain while protecting data privacy.
Li et al.\cite{li2024authenticated} worked on the verifiable keyword queries on graphs in blockchain-assisted cloud through aggregating invalid paths that are absent from the result trees.
Ge et al.\cite{ge2023privacy} proposed a scheme that supports querying for the subgraph same as the queried graph by using two servers without collusion.
Wang et al.\cite{wang2024secgraph} proposed a scheme named \emph{SecGraph} that enables privacy-preserving query and update on encrypted outsourced graph data by leveraging Intel SGX.
However, these works focus on cloud-based graph queries, while the on-chain graph query solution is still lacking.

\section{Preliminaries}

\textbf{Cryptographic Hash Function.}
A cryptographic hash function $H(\cdot)$ can map data of arbitrary size to a fixed-length string. It should be collisionless, i.e., given any two different messages $m_1,m_2$, the probability of $H(m_1)\text{=}H(m_2)$ is negligible.

\begin{figure}[t]
    \setlength{\abovecaptionskip}{0 cm}
    \setlength{\belowcaptionskip}{0 cm}
    \centering
    \includegraphics[width=0.98\textwidth]{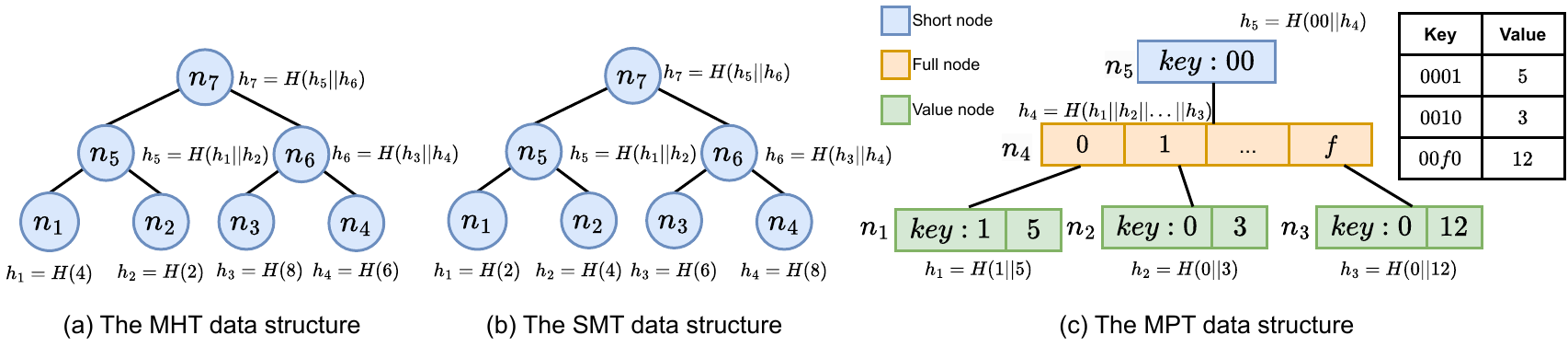}
    \caption{The data structures of MHT, SMT and MPT.}
    \label{fig:MHT}
    \vspace{-0.5 cm}
\end{figure}

\noindent\textbf{Merkle Hash Tree and Sorted Merkle Tree.}
Merkle hash tree (MHT)\cite{proceeding:MHT} is a classic ADS that could efficiently prove the existence of some objects.
Through building a binary hash tree, MHT authenticates multiple values using a single hash digest.
Fig.\ref{fig:MHT}(a) shows a MHT build on a set $S=\{2, 4, 6, 8\}$.
Each leaf node $n_i$ stores a value $v$ in $S$, while its hash digest $h_i$ is $H(v)$, e.g., $h_1=H(2)$.
\textcolor{black}{An internal node stores $h_l$ and $h_r$, the digests of the left and right child nodes, respectively. Its hash digest is defined as $H(h_l||h_r)$, e.g., $h_5=H(h_1||h_2)$.}
To prove the existence of value $2$, the merkle proof is generated as $\pi_{n_2}=\langle n_2, h_1, h_6 \rangle$.
Then the verifier rebuilds the root hash using $\pi_{n_2}$, i.e., $H(H(h_1||H(2))||h_6)$, and compares it with the root hash to check the existence of value $2$.
To efficiently prove the non-existence of objects,
Laurie et al.\cite{article:SMT} proposed a variant of MHT named sorted merkle tree (SMT). Unlike the traditional MHT, the leaf nodes in the SMT are sorted according to their values\footnote{If the values are non-numeric data, we use lexicographic sorting.}, shown in Fig.\ref{fig:MHT}(b). To prove the non-existence of value $5$, we only required to generate existence proofs for the two adjacent leaf nodes $n_2$ and $n_3$, which are exactly the boundaries of $5$, i.e., $\pi_{n_2}=\langle n_2, h_1, h_6\rangle$ and $\pi_{n_3}=\langle n_3, h_4,h_5\rangle$. After passing the verification of the existence of the two nodes, the verifier believes that the value $5$ does not exist.

\noindent\textbf{Merkle Patricia Trie.}
Merkle patricia trie (MPT) is an authenticated key-value store with MHT and patricia trie as the underlying data structures.
As depicted in Fig.\ref{fig:MHT}(c), each key is inserted into MPT according to its prefix path.
Similar to MHT, each node in MPT corresponds to a hash digest, e.g., $h_4=H(h_1||h_2||...|| h_3)$ and $h_5=H(00||h_4)$.
Thus, multiple nodes in the merkle path will be modified once updating a key.
For example, if the value for key '$00f0$' is updated, $h_3$, $h_4$ and $h_5$ will be updated recursively.
For simplicity, in this paper, we abstract the operations of MPT into three functions: \textsf{\footnotesize{Set}}$(k,v)$: it updates the value of key $k$ as $v$;
\textsf{\footnotesize{Get}}$(k') \rightarrow \langle v', \pi \rangle$: it outputs the value $v'$ and the merkle proof $\pi$ associated with key $k'$;
\textsf{\footnotesize{KVcheck}}$(H_m, k, v, \pi) \rightarrow 0/1$: it checks if $v$ is the associated value of key $k$ based on MPT's root hash $H_m$ and a merkle proof $\pi$. If correct, it outputs 1, otherwise 0.

\section{System Overview}


\noindent\textbf{Problem Definition.}
In this paper, we mainly focus on social graphs modeled as $G=\langle V,E \rangle$, where $V$ denotes the vertex set and $E$ denotes the edge set.
As shown in Fig.\ref{fig:system_overview}(a), each vertex represents a person with a unique identifier, and each edge in the form of $\langle type, weight \rangle$ describes the relationship between two persons.
For example, there is a 'friend'-typed edge between two vertexes '0a' and '08' with a weight of 5. This edge refers that the two persons with identifiers '0a' and '08' are friends exhibiting an intimacy of 5.
A top-$k$ graph query $Q$ = $(u_q,type_q,k)$ refers to retrieving $k$ edges with the highest weights from all edges of type $type_q$ incident to the vertex $u_q$.
For example, the top-$2$ closest friends of '0a' are '3f' and '08'. Table. \ref{tab:tab2} presents the notations used in this paper.


\noindent\textbf{System Model.}
We strive to achieve authenticated top-$k$ graph query over blockchain-stored graph data.
Analogous to vChain\cite{conference:vchain}, we model the transactions of the blockchain as objects.
An object is defined as $\langle u, v, type, w \rangle$ to present an edge, where $u$ is the start vertex while $v$ is the end vertex, $type$ is the type of the edge, and $w$ is the edge weight.
To ensure data integrity, an ADS is built over intra-block objects for each block of the blockchain. A block is identified by a unique incremental identifier $id$ (i.e., the block height) and denoted by $blk_{id}$.
Fig.\ref{fig:system_overview}(b) shows the layout of a block to store the objects (edges) in Fig.\ref{fig:system_overview}(a).
The system model of NetChain is composed of three parts: (i) Miner, a full node who is responsible for collecting transactions and packing them into new blocks. (ii) User, a light node that merely records the block headers acquired from the blockchain network. The User issues query requests and verifies the returned results. (iii) Service Provider (SP), a full node who maintains a copy of the whole blockchain. Upon receiving a query request from a User, SP provides search services in its local copy of the blockchain and returns search results with the verification object (VO). 

In this paper, we aim to execute top-$k$ graph queries over blocks whose $id$ is within a time window $[lb,ub]$.
Accordingly, a query request is defined as $Q=\{u_q, type_q, k, [lb,ub]\}$.
For the convenience of later descriptions, we denote the object having $u=u_q$ and $type=type_q$ as a matched object, and the block containing at least one matched object as a matched block.

\noindent\textbf{Threat Model.}
In our threat model, the User is considered as honest, and we assume that most of Miners work normally.
The SP is considered a potential adversary. In case of device failure or SP intends to conserve the computational resource, the SP may not execute the predefined procedure, and return tempered or incomplete search results to User. So it is necessary to devise an authenticated top-$k$ graph query framework to ensure 
\textcolor{black}{(i) \emph{Soundness:} all objects in the search result are real and exist in the blockchain; (ii) \emph{Correctness:} all objects in the search result satisfy the query condition, and haven't been tampered with, and (iii) \emph{Completeness:} no object satisfying the query is missing from the result.}

\begin{table}[t]\scriptsize
\begin{center}
\caption{Frequently-used notations.}
\label{tab:tab2}
\begin{tabular}{c|c|c|c}
\hline
\textbf{Notation} & \textbf{Definition} & \textbf{Notation} & \textbf{Definition}\\
\hline
$u,type,w$ & Vertex, edge type and edge weight & $o$ & Graph data object $\langle t,u,v,type,w \rangle$\\

$\mathcal{K}$ & Compound key $\langle u,type \rangle$&  $\mathcal{V}$ & Compound value $\langle v,w \rangle$ associated with $\mathcal{K}$\\
$C$ & Hash chain &  $ci$ & Hash chain item\\
$Q$ & Graph query &  $[lb,ub]$ & Query time window associated with $Q$\\
$\text{ADS}$ & Authenticated data structure &  VO & Verifiable object\\
$ptr,ptr_h$ & Hash pointer &  $blk_i$ & Block with height $i$\\
$r_i$ & Matched result in $blk_i$ &  $\pi_i$ & (Non-)existence proof for $blk_i$\\
$T_s$ & Sorted merkle tree (SMT) &  $H_s$ & Root hash of the $T_s$\\
$T_m$ & Merkle Patricia trie (MPT) &  $H_m$ & Root hash of the  $T_m$\\
\hline
\end{tabular}
\vspace{-3em}
\end{center}
\end{table}

\noindent\textbf{Design Goals.} 
In this paper, NetChain aims to achieve the following goals: (i) \emph{Authenticated top-k graph query.} NetChain should enable User to issue top-k graph queries towards SP and verify the returned search result; (ii) \emph{Efficient ADS generation.} NetChain should enable low generation time cost of ADS, and minimize its storage cost; (iii) \emph{Low communication cost.} NetChain should reduce the proof (i.e., VO) size largely, and (iv) \emph{Lightweight verification.} NetChain should support cost-efficient verification for Users, especially for resource-limited Users, such as smart mobile phones, wearable devices, etc.

\vspace{-1em}
\section{NetChain}

We first introduce the classic vChain framework\cite{conference:vchain} as a baseline solution and discuss its infeasibility when used in our scenario, from which we derive our design.
In vChain, a merkle hash tree (MHT) is built on all intra-block objects for each block of the blockchain as an ADS. Each MHT node contains $W$ and $AttDigest$.
For a leaf node, $W$ is a set of keywords in the object bind to it.
For an internal node, $W$ is the union of keywords sets of its child nodes.
Accordingly, $AttDigest$ is computed as the accumulative value of $W$, which is employed to generate non-existence proof for non-matched objects.
Given a query request, SP checks a block by traversing its MHT in a top-down fashion to handle (non-)matched objects.
For each traversed MHT node $n$, if its $W$ does not contain any matched keyword, SP stops searching, generates a non-existence proof based on $Attdigest$ and adds it into VO. The proof is used to prove that all the objects in the subtree rooted by $n$ mismatch the query condition.
Once reaching a leaf node that satisfies the query condition, SP will add the associated object into R, generate a merkle proof as the existence proof of the associated matched object and insert it into VO.
Then, R and VO are returned to User as the response.
Finally, User checks all the (non-)existence proofs in VO to verify the result.
 \begin{figure}[t]
    \centering
\includegraphics[width=0.8\textwidth]{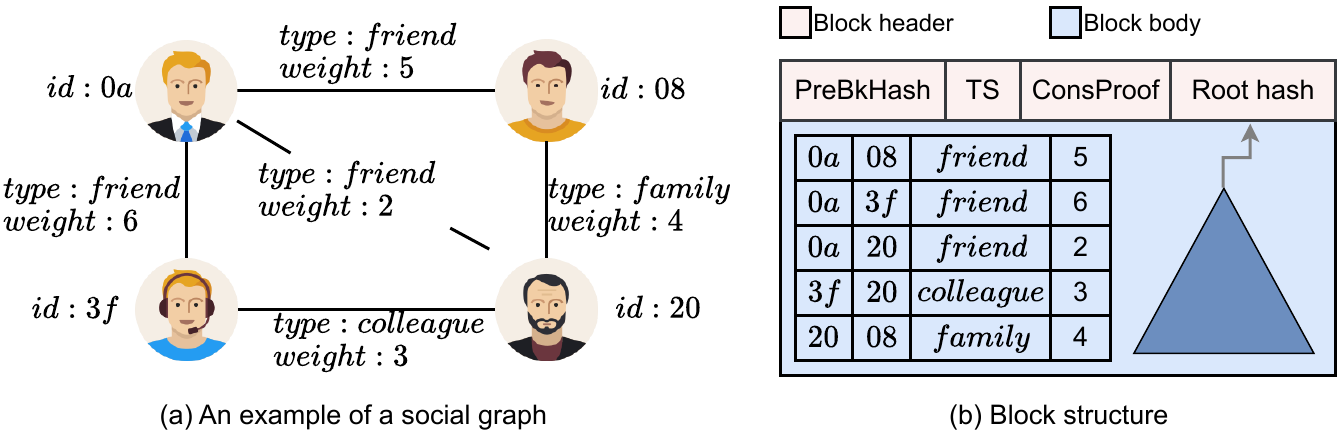}
    \caption{System overview.}
    \label{fig:system_overview}
    \vspace{-2.5em}
\end{figure}

\noindent \textbf{Design Overview.}
Apparently, vChain can be adapted to support authenticated top-$k$ graph query by returning all matched objects and letting User select top-$k$ results by itself.
However, vChain raises expensive overhead to process (non-)existence proofs.
First, the per-block ADS construction incurs expensive construction cost to deal with non-existence proofs.
Considering the MHT (ADS) of a block. Each MHT node has to store a keywords set $W$, leading to substantial storage cost for the block.
Meanwhile, the computation of the $Attdigest$ involves asymmetric cryptography operations, consuming lots of computation resources.
Second, the per-query VO construction spends significant authentication cost to deal with existence proofs.
Considering the matched objects of a query. The VO needs to include individual existence proofs for each of them and requires to be transmitted to User for verification, incurring expensive communication and verification costs.


We devise new ADSs to deal with (non-)existence proofs in a lightweight way. Considering a top-$k$ graph query $Q$.
Instead of proving (non-)existence for \mbox{(non-)matched} objects about $Q$, we choose to prove (non-)existence for \mbox{(non-)matched} blocks, where a matched block contains matched objects and vice versa. Given $Q=\{u_q, type_q, k, [lb,ub]\}$, our observation is that all matched objects must fulfill the query condition $u=u_q$ and $type=type_q$.
As a result, we can divide all the objects within a block into clusters with each one indexed by a unique compound key $\mathcal{K}$ = $\langle u,type \rangle$.
We then use the sorted merkle tree (SMT) as the underlying data structure to construct a \emph{compound key-based SMT} over these compound keys for the block.
In this way, SP can easily prove the block as a (non-)matched one by generating a (non-)existence proof to prove (non-)existence of a compound key $\mathcal{K}=\langle u_q,type_q \rangle$.
Compared with vChain, the compound-key based SMT can be constructed efficiently since it only involves hash functions, and occupies less storage since each SMT node merely retains the hash digest of its child nodes.

For a matched block about $Q$, we further remove the need of generating individual existence proofs for matched objects within the block.
Our idea is to aggregate objects indexed by the same compound key as an \emph{ordered hash chain} with built-in verifiability.
For a compound key $\mathcal{K}$ in a SMT leaf node $n$, objects with the same $\mathcal{K}$ are organized and sorted by their weights in descending order to form a hash chain. Each chain item is composed of an object and a hash pointer.
To preserve the data verifiability, the hash pointer in each chain item is computed as the hash digest of its successor, while the hash pointer of the head chain item is embedded into $n$, which can be verified along with the merkle proof. In this way, the verifiability of hash items are rooted to the verifiability of the leaf node.
As a result, SP can solely return top-$k$ chain items as matched objects, and User can directly verify them via hash chain re-computation.

Finally, we build an authenticated two-layer index by combining compound key-based SMT with ordered hash chains as an ADS of a block and stored in the block body.
For example, there are four compound keys for objects of a block as shown in Fig.\ref{fig:design_overview}(a), and a compound key-based SMT is built on them as shown in Fig.\ref{fig:design_overview}(b).  Considering a top-$k$ graph query $Q$.
In case that the query condition $\mathcal{K}_q=\langle u_1,t_3 \rangle$, the block is a non-match one and SP returns a non-existence proof $\{\langle n_2,h_1,h_6 \rangle, \langle n_3,h_4,h_5 \rangle \}$ to User.
As shown in Fig.\ref{fig:design_overview}(c), in case that $\mathcal{K}_q=\langle u_1,t_1 \rangle$, the block is a match one. SP returns the merkle proof of $n_1$ with $\mathcal{K}=\langle u_1,t_1 \rangle$ as an existence proof as well as top-$k$ hash items of $\mathcal{K}$. Compared with vChain, our design achieves lightweight ADS construction and deals with (non-)existence proofs in block-level.

\begin{figure}[t]
    \centering
    \includegraphics[width=0.98\linewidth]{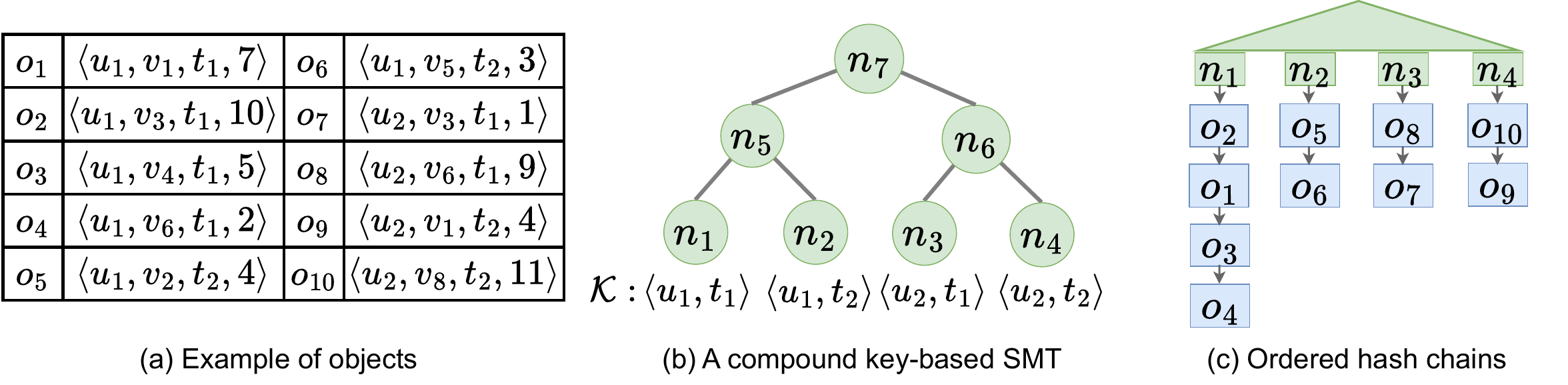}
    \caption{The ADS designs.}
    \label{fig:design_overview}
    \vspace{-2.5em}
\end{figure}

\noindent \textbf{Authenticated Two-layer Index.}
Now, we discuss the two-layer authenticated index construction and detail it in Alg.\ref{alg:netchain}.
Fig.\ref{fig:ADS} presents an example based on the case in Fig.\ref{fig:design_overview}.
We first describe the ordered hash chain construction.
In practice, due to the shared compound key $\mathcal{K}$, we simplify the chain item $ci$ to store a compound value $\mathcal{V}=\langle v,w \rangle$ and a hash pointer $ptr$ (lines 4-13), where the $ptr$ is computed by invoking \textsf{\footnotesize{Digest}}($\cdot$) defined in lines 53-54.
Particularly, $ptr$ in the last chain item is set as $\perp$.
As for SMT, the definitions of leaf nodes and internal nodes are different.
A leaf node $n$ has two fields: compound key $\mathcal{K}$, and hash pointer $ptr_h$ which is hash digest of the first chain item.
An internal node contains the child hashes, i.e., $hash_l$ and $hash_r$.
The hash digests of leaf node and internal node are computed as $H(\mathcal{K}||ptr_h)$ and $H(hash_l||hash_r)$ respectively.
Upon the SMT $T_s$ is built, its root hash $H_s$ is embedded into block header (line 16).

\begin{figure}[t]
    \centering
    \includegraphics[width=0.8\linewidth]{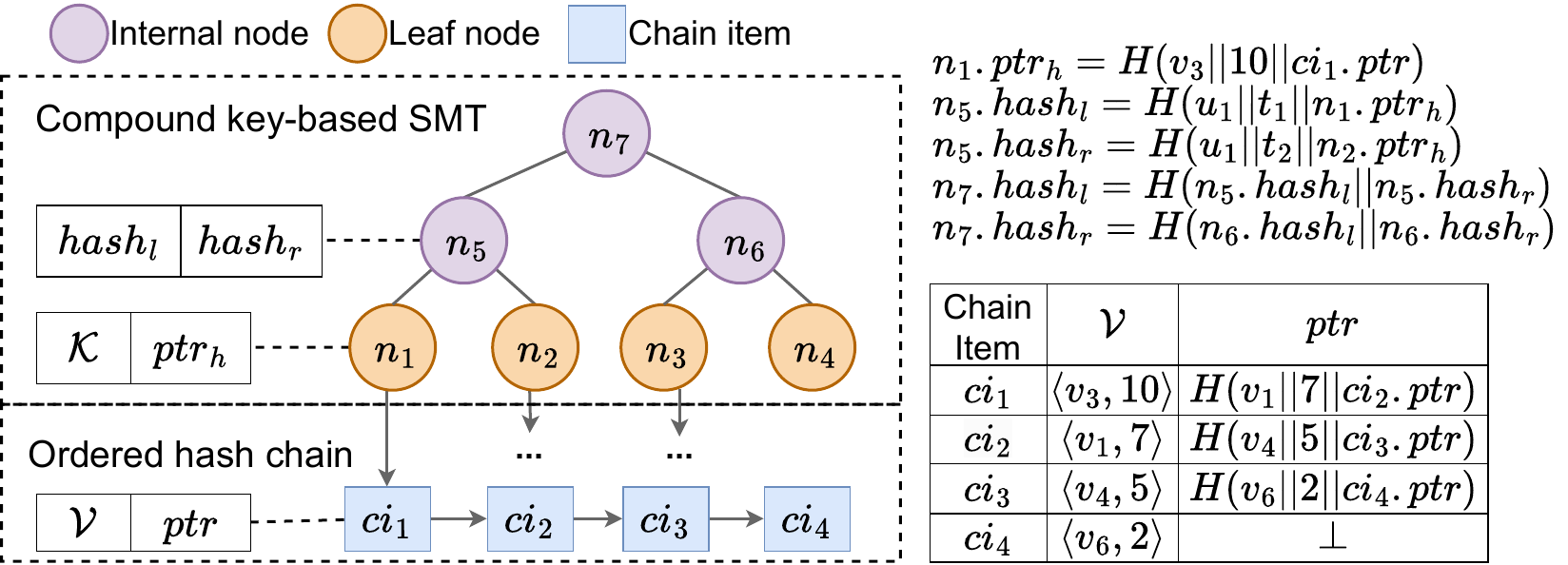}
    \caption{Authenticated two-layer index.}
    \label{fig:ADS}
\end{figure}    
    
\setlength{\textfloatsep}{0.1em}    
\begin{algorithm}[t]\scriptsize    
\SetNlSty{text}{}{:}       
\SetAlgoNoLine
\SetAlgoNoEnd

    \caption{NetChain}
    \label{alg:netchain}

    \begin{multicols}{2}
    
    \underline{\textbf{ADS Construction}:}\\
    \nl Init a set $S$ and a dictionary $D$\;
    \nl \textbf{for} each $\mathcal{K}$ in a transactions pool \textbf{do} \\
    \nl ~~~~Init an empty hash chain $C$\;
    \nl ~~~~\textbf{for} each object $o$ containing $\mathcal{K}$ \textbf{do} \\
    \nl ~~~~~~~~~~Extract $\mathcal{V}=\langle v,w \rangle$ from $o$\;
    \nl ~~~~~~~~~~Insert $(\mathcal{V},\perp)$ into $C$\;
        
    \nl ~~~~~Sort $C$ by $w$ in descending order\;
    \nl ~~~~~$l\leftarrow$ length of $C$, $j\leftarrow l-1$\;
    \nl ~~~~~\textbf{while} $j > 0$ \\
    \nl ~~~~~~~~~~$C[j].ptr$ $\leftarrow$ \textsf{Digest}($C[j+1]$)\; 
    \nl ~~~~~~~~~~$j\leftarrow j-1$\;

    \nl ~~~~~$D[\mathcal{K}] \leftarrow C$\; 
    \nl ~~~~~Add ($\mathcal{K}, \textsf{Digest}(C[1]$)) into $S$\;

    \nl Construct a SMT $T_s$ using leaf nodes in $S$\;
    \nl Write $T_s$ and $D$ into the block body\;
    \nl Set $H_{s}$ in block header as the root hash of $T_s$\;    
    
    \textbf{\underline{Search:}}\\

    \nl Parse query request $Q$ as $\langle u_q, type_q, k, [lb, ub] \rangle $;\\
    
    \nl $\mathcal{K}_q \leftarrow \langle u_q,type_q \rangle$, R$\leftarrow \emptyset$, VO$\leftarrow \emptyset$\;
    
    \nl \textbf{for} each block $blk_i$ with $i$ in $[lb,ub]$ \textbf{do}\\
    \nl ~~~~~$T_s, D \leftarrow blk_i$\; 
    \nl ~~~~~\textbf{if} $\mathcal{K}_q\in T_s$ \textbf{then} \\
    \nl ~~~~~~~~~~$\pi_i \leftarrow \textsf{ProExistence}(T_s,\mathcal{K}_q$)\;
    \nl ~~~~~~~~~~$C\leftarrow D[\mathcal{K}_q]$, $l\leftarrow$ length of $C$\;
    \nl ~~~~~~~~~~$r_i\leftarrow  C[1:\textsf{min}(l,k)]$\;
    \nl ~~~~~\textbf{else} \\
    \nl ~~~~~~~~~~$\pi_i \leftarrow \textsf{ProNonExistence} (T_s,\mathcal{K}_q)$, $r_i\leftarrow \perp$\;

    \nl ~~~~VO $\leftarrow \text{VO} \cup \pi_i$, $\text{R}\leftarrow \text{R}\cup r_i$\;
    \nl Send R,VO to the User\;

    \underline{\textbf{Verify:}}\\
    \nl $Res\leftarrow \emptyset$\;
    
    \nl \textbf{for} each $i$ in $[lb,ub]$ \textbf{do} \\
    \nl ~~~~~Retrieve $H_s$ from the block header of $blk_i$\;
    \nl ~~~~~$r_i \leftarrow \text{R}[i], \pi_i \leftarrow \text{VO}[i]$\;
    \nl ~~~~~\textbf{if} $r_i\neq \perp$ \textbf{then} \\
    \nl ~~~~~~~~~~\textbf{if} \rm{\textsf{VerExistence}}($H_s,\mathcal{K}_q, \pi$) $\neq 1$ \textbf{then}\\
    \nl ~~~~~~~~~~~~~~~Exit and report error\;
            
    \nl ~~~~~~~~~~Get $\mathcal{K},ptr_h$ from SMT leaf node in $\pi$\;
    \nl ~~~~~~~~~~\textbf{if} $\mathcal{K} \neq \mathcal{K}_q$ \textbf{then} \\
    \nl ~~~~~~~~~~~~~~~Exit and report $error$\;
    \nl ~~~~~~~~~~$l\leftarrow$ length of $r_i$\;
    \nl ~~~~~~~~~~\textbf{if} $l < k$ \rm{and} $r_i[l].ptr\neq \perp$ \textbf{then} \\
    \nl ~~~~~~~~~~~~~~~Exit and report $error$\;

    \nl ~~~~~~~~~~$j\leftarrow 1,ptr_{pre}\leftarrow ptr_h$\;
    \nl ~~~~~~~~~~\textbf{while} $j\leq l$ 
    \textbf{do} \\
    
    \nl ~~~~~~~~~~~~~~~$\mathcal{V},ptr \leftarrow r_i[j]$, $h\leftarrow \textsf{Digest} (r_i[j])$\;
    \nl ~~~~~~~~~~~~~~~\textbf{if} $h \neq ptr_{pre}$ \textbf{then} \\
    \nl ~~~~~~~~~~~~~~~~~~~~Exit and report $error$\;
            
    \nl ~~~~~~~~~~~~~~~$ptr_{pre}\leftarrow ptr$, $j\leftarrow j+1$\;
    \nl ~~~~~~~~~~~~~~~$Res\leftarrow Res \cup \mathcal{V}$\;
    \nl ~~~~~\textbf{else}\\
    \nl ~~~~~~~~~~\textbf{if} \rm{\textsf{VerNonExistence}}$(H_s,\mathcal{K}_q, \pi_i)\neq 1$ \textbf{then} \\
    \nl ~~~~~~~~~~~~~~~Exit and report error; \\

    \nl Select the global top-$k$ results from $Res$\;

\textbf{Function} \textsf{Digest}($ci$) \textbf{:} \tcc{$ci$ is a chain item}

\nl ~~~~~$h \leftarrow H(ci.\mathcal{V}.v || ci.\mathcal{V}.w || ci.ptr)$\;
\nl ~~~~~return $h$;\\

    \end{multicols}
\end{algorithm}

\noindent \textbf{Verifiable Query Process.}
Consider a query s.t., $Q=\langle u_q, type_q, k, [lb, ub] \rangle$.
Initially, SP generates two sets R and VO.
For each block $blk_i$ in time window, SP retrieves the SMT $T_s$ from block body, and determines the existence of $\mathcal{K}_q=\langle u_q,type_q \rangle$ (lines 20-21).
If $blk_i$ is a matched block, SP first generates existence proof $\pi_i$ for $\mathcal{K}_q$ (line 22).
Subsequently, SP accesses the ordered hash chain $C$ related to $\mathcal{K}_q$ and partitions top-$k$ chain items (i.e., $C[1:k]$) as $r_i$ (lines 23-24).
Otherwise, when $blk_i$ is a non-matched block, SP generates the non-existence proof $\pi_i$ for $\mathcal{K}_q$ and set $r_i$ as $\perp$ (line 26).
Finally, $r_i$ and $\pi_i$ are inserted into R and VO respectively (line 27).
When finishing the processes for all blocks in the time window, R and VO are returned to User as the response.
For simplicity, here we omit the analysis on case that the count of intra-block matched objects is less than $k$, but it is easy to fix as shown in line 24.


In verification stage, for each $i$ in time window, User first retrieves $H_s$ from $blk_i$'s block header (line 31).
In the case that $blk_i$ is claimed as a matched block (line 33), User first verifies the existence of $\mathcal{K}_q$ (lines 34-35).
Meanwhile, User extracts $\mathcal{K}$ from the SMT leaf node in $\pi_i$ and checks if $\mathcal{K}=\mathcal{K}_q$ (lines 36-38).
After that, User checks the $r_i$ integrity (lines 39-48).
For each chain item $r_i[j]$ in $r_i$, User iteratively computes its hash digest and compares it with $r_i[j-1].ptr$ (specially, integrity of $r_i[1]$ is checked using $ptr_h$ in SMT leaf node from $\pi_i$).
As for the case that $blk_i$ is claimed as a non-matched block, User only verifies the non-existence of $\mathcal{K}_q$ (lines 49-51).
When the response passes verification, User selects the global top-$k$ objects from R as results.


\emph{Example.}
Fig.\ref{fig:example1} shows a blockchain with hundreds of blocks following the example in Fig.\ref{fig:design_overview}, while the ADS of $blk_5$ is presented in Fig.\ref{fig:ADS}.
Assume the query condition is $\{u_1,t_1, 3, 0, 299 \}$ and denote $\mathcal{K}_1$ as $\langle u_1,t_1 \rangle$.
There are three matched blocks in time window: $blk_5,blk_{73},blk_{219}$.
Corresponding $r_i$ of these matched blocks are marked by red box.
Consider the authentication for $blk_5$.
The proof $\pi_5$ is set as $\{n_1, h_2, h_6 \}$, and $r_5$ is set as $\{ci_1, ci_2, ci_3\}$.
In verification stage, for $blk_5$, User reconstructs the SMT root hash s.t. $H((H(n_1)||h_2)||h_6)$, and compares it with $H_s$.
After that, User gets $ptr_h$ in $n_1$ from $\pi_5$, and compares it with \textsf{\footnotesize{Digest}}($ci_1$) to ensure the integrity of $ci_1$.
Then, User iteratively verifies the remaining chain items using embedded $ptr$.
Finally, User obtains global top-$3$ result as $\{ \langle v_2,18 \rangle,\langle v_8,15 \rangle, \langle v_{13},11 \rangle \}$.


\begin{figure}[t]
    \centering
    \includegraphics[width=0.8\linewidth]{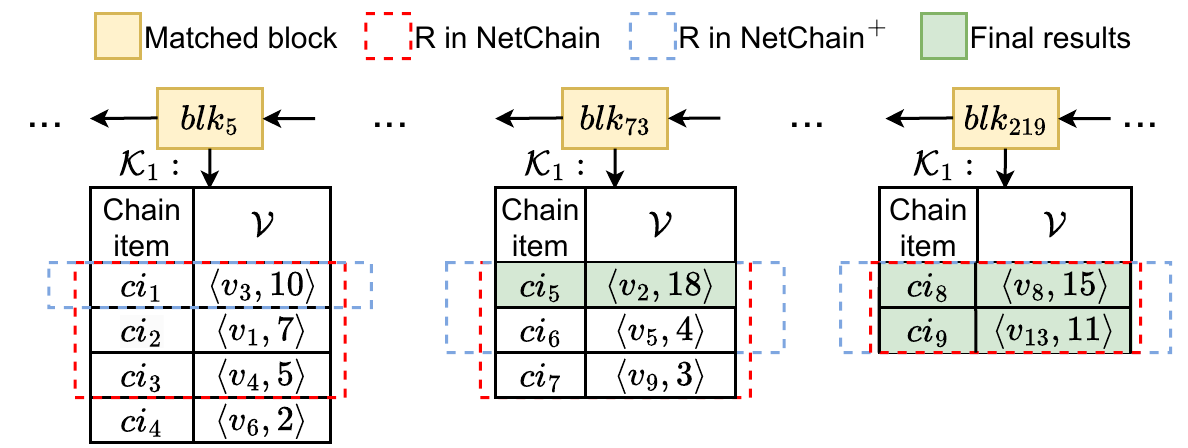}
    \caption{Global data distribution.}
    \label{fig:example1}
\end{figure}

\section{NetChain$^{+}$: Optimized Authentication Cost}
NetChain achieves authenticated top-$k$ graph query on blockchain.
However, it still suffers from high authentication cost (i.e., communication and verification cost) for two reasons. Considering a query $Q=\langle \mathcal{K}_q,k,[lb,ub] \rangle$. Since many blocks within [$lb$, $ub$] could be non-matched ones, SP have to generate and add many non-existence proofs into VO for User to verify.
In addition, SP inserts top-$k$ objects of each matched block into R, causing the size of R much larger than $k$.
Thus, we attempt to minimize communication and verification cost of NetChain.
In this section, we propose NetChain$^{+}$, where we devise (i) a verifiable inter-block link construction to reduce the authentication cost for VO, and (ii) a two-round scan mechanism to reduce the authentication cost for R.

\setlength{\textfloatsep}{0.1em}    
\begin{algorithm}[t]\scriptsize    
\SetNlSty{text}{}{:}       
\SetAlgoNoEnd       

    \caption{NetChain$^{+}$}
    \label{alg:optimization}

    \begin{multicols}{2}
    \underline{\textbf{ADS Construction:}}\\
    \nl Init a set $S$ and a dictionary $D$\;
    \nl \textbf{for} each $\mathcal{K}$ in a transactions pool \textbf{do}\\
        \nl ~~~~ Lines 3-12 in Alg.\ref{alg:netchain}, $id_{pre}\leftarrow -1$\;
        \nl ~~~~ \textbf{if} $\mathcal{K}\in T_m$ \textbf{then} \\
        \nl ~~~~~~~~ $id_{pre},\pi_{mpt} \leftarrow T_m$.\textsf{Get}($\mathcal{K}$)\;
        
        \nl ~~~~ Add $(\mathcal{K}, \textsf{Digest}(C[1]),id_{pre})$ into $S$\;
        \nl ~~~~ $T_m$.\textsf{Set}($\mathcal{K}$, current block $id$)\;

    \nl Lines 14-16 in Alg.\ref{alg:netchain}\;
    \nl Set $H_m$ in block header as the root hash of $T_m$\;

    \underline{\textbf{Search:}}\\
    \nl Lines 17-18 in Alg.\ref{alg:netchain}\;
    \nl $i,\pi_{mpt} \leftarrow T_m$.\textsf{Get}$(\mathcal{K})$\;
    \nl \textbf{if} $i \leq ub$ \textbf{then}\\
        \nl ~~~~ $a \leftarrow i, b\leftarrow i$, $\text{VO}\leftarrow \text{VO}\cup \pi_{mpt}$\;
        \nl \textbf{else}\\
        \nl ~~~~ \textbf{while} $i>ub$ \textbf{do}\\
            \nl ~~~~~~~~ $T_s \leftarrow blk_{i}$\;
            \nl ~~~~~~~~ $n \leftarrow$ leaf node in $T_s$ containing $\mathcal{K}_q$\;
            \nl ~~~~~~~~ \textbf{if} $n.id_{pre} \leq ub$ \textbf{then}\\
                \nl ~~~~~~~~~~~~ $a\leftarrow n.id_{pre}, b\leftarrow i$\;
                \nl ~~~~~~~~~~~~ $\pi_b \leftarrow$ \textsf{ProExistence}($T_s,\mathcal{K}_q$)\;
                \nl ~~~~~~~~~~~~ $\text{VO}\leftarrow \text{VO} \cup \pi_b$, break\;
            
            \nl ~~~~~~~~ $i\leftarrow n.id_{pre}$\;

    \nl Init a set $I$, $i\leftarrow a$\;
    \nl \textbf{while} $i$ \rm{in} $[lb,ub]$ \textbf{do}\\
        \nl ~~~~ $T_s,D \leftarrow blk_i$, $C\leftarrow D[\mathcal{K}_q]$\;
        \nl ~~~~ $n\leftarrow$ leaf node in $T_s$ containing $\mathcal{K}_q$\;
        \nl ~~~~ Add all $\mathcal{V}$ in $C$ into $I$, $i \leftarrow n.id_{pre}$\;
    
    \nl Select top-$k$ $\mathcal{V}$ from $I$ as R$_Q$, $i\leftarrow a$\;

    \nl \textbf{while} $i$ \rm{in} $[lb,ub]$ \textbf{do}\\
        \nl ~~~~ $T_s,D \leftarrow blk_i$\;
        \nl ~~~~ Lines 22-23 in Alg.\ref{alg:netchain}, $j\leftarrow 1$\;
        \nl ~~~~ \textbf{while} $j\leq l$ \textbf{do}\\
            \nl ~~~~~~~~ \textbf{if} $C[j].\mathcal{V}\notin \text{R}_Q$ or $j=l$ \textbf{then}\\
                \nl ~~~~~~~~~~~~ $r_i \leftarrow C[1:j]$, break\;
            
            \nl ~~~~~~~~ $j\leftarrow j+1$\;
        
        \nl ~~~~ Line 27 in Alg.\ref{alg:netchain}\;
        \nl ~~~~ $n\leftarrow$ leaf node in $T_s$ containing $\mathcal{K}_q$\;
        \nl ~~~~ $i \leftarrow n.id_{pre}$\;
    
    \nl Send $\text{R}, \text{VO}, b$ to the User\;

    \underline{\textbf{Verify:}}\\

    \nl Find top-$k$ $\mathcal{V}$ from R as $Res$\;

    \nl \textbf{if} $b \leq ub$ \textbf{then}\\
        \nl ~~~~ Get $\pi_{mpt}$ from VO\;
        \nl ~~~~ Retrieve $H_{m}$ from latest block header\;
        \nl ~~~~ \textbf{if} \textsf{KVcheck}$(H_m, \mathcal{K}_q, b, \pi_{mpt}) \neq 1$ \textbf{then}\\
            \nl ~~~~~~~~ Exit and report \emph{error}\;
        \nl ~~~~ $a\leftarrow b$\;
    \nl \textbf{else}\\
    
        \nl ~~~~ Get leaf node $n$ in $\pi_b$\;
        \nl ~~~~ Retrieve $H_s$ from the block header of $blk_b$\;
        \nl ~~~~ \textbf{if} \textsf{VerExistence}$(H_s, \mathcal{K}_q,\pi_b)\neq 1$ or $n.id_{pre}>ub$ \textbf{then}\\
        \nl ~~~~~~~~Exit and report \emph{error}\;

        \nl ~~~~ $a\leftarrow n.id_{pre}$\;
    
    \nl $i\leftarrow a$\;
    \nl \textbf{while} $i$ in $[lb,ub]$ \textbf{do}\\             
        \nl ~~~~ Lines 31-32 in Alg.\ref{alg:netchain}\;
        \nl ~~~~ \textbf{if} \textsf{VerExistence}$(H_s,\mathcal{K}_q,\pi_i)\neq 1$ \textbf{then}\\
        \nl ~~~~~~~~ Exit and report \emph{error}\;
        
        \nl ~~~~ Get $\mathcal{K},ptr_h,id_{pre}$ from leaf node in $\pi_i$\;
        
        \nl ~~~~ Lines 37-38 in Alg.\ref{alg:netchain}\;
        \nl ~~~~ $j\leftarrow 1$, $ptr_{pre}\leftarrow ptr_h$, $l\leftarrow$ length of $r_i$\;

        \nl ~~~~ \textbf{while} $j\leq l$ \textbf{do}\\
            \nl ~~~~~~~~ Lines 44-46 in Alg.\ref{alg:netchain}\;

            \nl ~~~~~~~~ $flag_1 \leftarrow j<l'$ and $\mathcal{V} \in Res$\;
            \nl ~~~~~~~~ $flag_2 \leftarrow j=l' $ and $ \mathcal{V}\notin Res$\;
            \nl ~~~~~~~~ $flag_3 \leftarrow j=l' $ and $ \mathcal{V} \in Res$ and $ptr=\perp$\;
            \nl ~~~~~~~~ \textbf{if} $flag_1 $ \rm{or} $ flag_2 $ or $ flag_3 \neq 1$ \textbf{then} \\
                \nl ~~~~~~~~~~~~ Exit and report $error$\;
            
            \nl ~~~~~~~~ $ptr_{pre}\leftarrow ptr$, $j\leftarrow j+1$\;
        \nl ~~~~ $i\leftarrow id_{pre}$\;

    \nl Take $Res$ as the final top-$k$ results\;
    
    \end{multicols}
\end{algorithm}

\noindent{\underline{\emph{Verifiable Inter-block Link Construction.}}}
Since NetChain lacks aggregation across blocks, SP has to generate (non-)existence proof per block, incurring significant authentication overhead for many non-matched blocks.
To address this issue, we propose a verifiable inter-block link construction to eliminate the authentication for non-matched blocks.
Our idea is to link blocks containing the same compound key in a verifiable way.
As a result, SP only needs to generate existence proofs for matched blocks within [$lb$, $ub$] and adds them into VO. On the other hand, User can verify 
that existence proofs of all the matched blocks are correctly returned via the linkage of $\mathcal{K}_q$ among them.

To establish linkages among blocks, considering a newly added block $blk_i$.  For each  $\mathcal{K}$ in $blk_i$, we first retrieve the preceding block $blk_j$ containing $\mathcal{K}$, and establishes an inter-block link from $blk_i$ to $blk_j$.
To do so, $j$ is embedded into the SMT leaf node bind to $\mathcal{K}$ in $blk_i$ as a link.
For example, as shown in Fig.\ref{fig:example2}, there is a link between $blk_{219}$ and $blk_{73}$ due to the common owned $\mathcal{K}_1$, and 73 is stored in the corresponding SMT leaf node in $blk_{219}$.
We keep a global merkle Patricia trie (MPT) $T_m$ within the blockchain system (maintained by all participants except the User) as an authenticated key-value store, where each key-value pair contains a compound key $\mathcal{K}$ and $id$ of the block that $\mathcal{K}$ occurs most recently.
E.g, the value for $\mathcal{K}_1$ in MPT is updated to 301 once $blk_{301}$ is newly added.
In this way, SP can directly retrieve all matched blocks fulfilling $\mathcal{K}_q$ within [$lb$, $ub$], and generates existence proofs for them.
For example, as shown in Fig.\ref{fig:example2}, if the query condition is $\mathcal{K}_1$, SP first searches MPT and learns that the latest matched block is $blk_{301}$. 
Starting from $blk_{301}$, SP can quickly traverse all the other matched blocks within the time window via links of $\mathcal{K}_1$ among them, and generates existence proofs as $\{\pi_{219},\pi_{73}, \pi_{5}\}$.
Moreover, to guarantee the result \emph{completeness}, the existence proof for the out boundary matched block is also added into VO, e.g., $\pi_{301}$.
Obviously, this design eliminates the authentication cost for non-matched blocks.

\begin{figure}[t]
    \centering
    \includegraphics[width=0.6\linewidth]{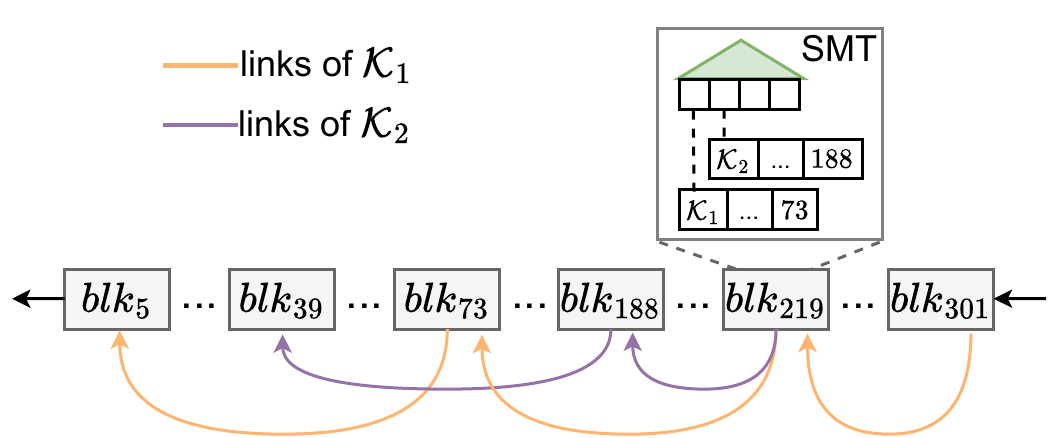}
    \caption{The inter-block index.}
    \label{fig:example2}
\end{figure}

\noindent{\underline{\emph{Two-round Scan Mechanism.}}}
Assume there are $N$ matched blocks in the time window [$lb$, $ub$].
The original design of Netchain requires SP to return local top-$k$ chain items for each of the $N$ matched blocks, leading to $O(N \times k)$ size of R, while the optimal size of R in ideal should be $O(k)$.
To address this issue, we propose a two-round scan mechanism, which returns the global top-$k$ chain items with a small set of $N$ additional chain items to achieve approximately $O(k)$ size of R.
In detail, SP first scans over all the matched blocks with $id$ in [$lb$, $ub$], and obtains the global top-$k$ chain items (1st round of scan). 
With the $k$ chain items, SP scans over all these matched blocks again to record how many valid chain items each matched block contains (2nd round of scan), where we define the valid chain item as the one exists in global top-$k$ results.
For example, as shown in Fig.\ref{fig:example1}, given the query condition $\mathcal{K}_1$, $blk_5$, $blk_{73}$ and $blk_{219}$ contributes 0, 1 and 2 valid chain items (marked by green color) to the global top-$k$ results. 


Apart from the integrity of chain items, User also needs to ensure that none of valid chain items is missed.
Directly returning all valid chain items will break the result \emph{completeness} since a malicious SP may miss some valid chain items, e.g., only returning $ci_8$ as the valid chain item $blk_{219}$ contains in Fig.\ref{fig:example1}.
Therefore, we need to return the successor chain item after those valid chain items additionally per block as the out boundary chain item, e.g., $ci_6$ in $blk_{73}$ is the out boundary chain item when querying for $\mathcal{K}_1$.

\noindent \textbf{Index Construction.} The detail of NetChain$^+$ is shown in Alg. \ref{alg:optimization}.
Denote $id_{pre}$ as $id$ of the preceding block containing the same $\mathcal{K}$.
The index construction procedure of NetChain$^{+}$ is basically the same with that of NetChain except for the inter-block links construction.
For each $\mathcal{K}$ in $blk_i$, Miner retrieves the local $T_m$ for $\mathcal{K}$ and gets the value $i'$ (lines 4-5).
Then Miner generates a leaf node associated with $\mathcal{K}$, the $id_{pre}$ of which is set as $i'$ (line 6).
After that, the value for $\mathcal{K}$ in $T_m$ is updated as $i$ (line 7).
Finally, the up to date $H_m$ is embedded into block header (line 9).

\noindent \textbf{Query Process.}
Denote $blk_a$, $blk_b$ as the right boundary chain item and the out boundary chain item. For example, when quering for $\mathcal{K}_1$ in the time window $[0,300]$, $blk_{219}$ and $blk_{301}$ are the right boundary and out boundary chain item, respectively.
In search stage, SP first determines the right boundary block $blk_a$ and the out boundary block $blk_b$, and generates proof $\pi_b$/$\pi_{mpt}$ (lines 11-22). 
Next, SP executes the first round of scan and obtains global top-k chain items R$_Q$ (lines 24-28). 
Then the second round of scan is executed. 
For each matched block, SP combines all valid chain items with the out boundary chain item as $r_i$ (lines 32-35), and generates the existence proof $\pi_i$. 
Finally, SP sends $\{\text{R,VO},b\}$ to User. 

In verification stage, User initially obtains top-$k$ results $Res$ from R. 
Next, User verifies the correctness of $b$.
In case that $b\leq ub$, User accesses the latest block header, from which gets $H_m$. 
Then User uses $H_m$ and $\pi_{mpt}$ to verify if the value for $\mathcal{K}_q$ in $T_m$ is actually $b$ (lines 44-45).
As for the case that $b>ub$, User gets $H_s$ from block header of $blk_b$, then simply verifies that $\mathcal{K}_q$ is indeed in $blk_b$ using $H_s$ and $\pi_b$.
Meanwhile, User makes sure that the preceding matched block of $blk_b$ is the right boundary block (lines 47-51). 
Once the verification for $b$ successes, User verifies R and VO. 
Initiate $i$ as $a$. 
In each iteration, User uses $\pi_i$ to verify the existence of $\mathcal{K}_q$ (lines 56-57).
Subsequently, $r_i$ is examined, including the verification for chain items integrity and validity (lines 63-67).
Note that $flag_1, flag_2, flag_3$ correspond to the three types of legal chain items. 
Finally, update $i$ as $id_{pre}$ embedded in the leaf node in $\pi_i$ (line 69).
This process recurs until $i<lb$.
If all verification success, User takes $Res$ as the correct final results.

\section{Security Analysis}


\begin{theorem}\label{thm1}
Following the definition of Unforgeability in \cite{conference:vchain}, NetChain satisfies the Unforgeability security property.
\end{theorem}

\begin{proof}
Before the proof, we let $\{o_i\}$ denote the set of real objects in the blockchain and the $\text{R}_Q$ refer to the genuine result. This theorem is proved by contradiction. 


\underline{Case 1 (Soundness): R contains an object $o'$ s.t. $o'\notin \{o_i\}$.}
Upon receiving R and VO, the verifier first examines the proofs in VO, then checks the matched objects in R.
There may be two possible conditions in this case:
(i) An adversary $\mathcal{A}$ successfully proves $o'$ exists in a non-matched block, indicating that a forged merkle proof is generated by finding the collision of cryptographic hash function.
(ii) $\mathcal{A}$ successfully replaces a matched object in R with a fake one, meaning that $\mathcal{A}$ can find a collision of hash function.
In summary, this case contradicts the security of cryptographic hash function.

\underline{Case 2 (Correctness): R contains an object $o'$ s.t. $o'\in \{o_i\}$, but $o'\notin \text{R}_Q$.} Apparently, this case is impossible since the verifier would check all objects in R, and filter the final results by itself.

\underline{Case 3 (Completeness): There is an object $o'\in \text{R}_Q$ in $\{o_i\}$, but $o'\notin $R.} There may be two possible conditions in this case:
(i) $\mathcal{A}$ successfully generates non-existence proof for a matched block.
(ii) $\mathcal{A}$ successfully proves an ordered hash chain $C$ contains less then $k$ objects, where $C$ has at least $k$ objects in fact.
Both of these conditions indicating that $\mathcal{A}$ finds a collision of cryptographic hash function, which leads to a contradiction.
\end{proof}

\begin{theorem}
NetChain framework satisfies the Unforgeability security property.
\end{theorem}

\begin{proof}
NetChain$^{+}$ is derived from NetChain, and the authentication method for \emph{soundness} and \emph{correctness} is similar to that of NetChain. Therefore, the Case 1 and Case 2 in NetChain$^{+}$ is infeasible due to the same reasons in Theorem \ref{thm1}.

\underline{Case 3 (Completeness): There is an object $o'\in \text{R}_Q$ in $\{o_i\}$, but $o'\notin $R.}
There are three conditions in this case:
(i) $\mathcal{A}$ successfully generates non-existence proof for a matched block.
(ii) $\mathcal{A}$ misses $o'$ from corresponding R.
(iii) $\mathcal{A}$ takes $o'$ as the out boundary object, while $o'$ is actually a valid object.
The first two conditions means that $\mathcal{A}$ finds a collision of cryptographic hash function, which leads to a contradiction.
The third condition is impossible since the verifier will filter the final results from R, where it is revealed that the weight of $o'$' is larger than that of some claimed valid objects, causing a contradiction.
\end{proof}

\begin{table}[t]
\scriptsize
\setlength{\abovecaptionskip}{0.cm}
	\caption{Dataset statistics.}
	\label{tab:dataset}
	\begin{tabular}{c|c|c|c|c|c}
		\hline 
		\textbf{Dataset}& \textbf{Nodes}  & \textbf{Edges} & \textbf{Edge type} & \textbf{Graph type} & \textbf{Source link} \\
		\hline
        Email & 36,692 & 183,831 & Friendship & Undirected & snap.stanford.edu/data/email-Enron.html\\
        Wiki & 7,115 & 103,689 & Voting & Directed & snap.stanford.edu/data/wiki-Vote.html\\
        GPlus & 107,614 & 13,673,453 & Share & Directed & snap.stanford.edu/data/ego-Gplus.html\\
		\hline
	\end{tabular}
\end{table}
\begin{table}[t]
    \setlength{\abovecaptionskip}{0pt} 
    \setlength{\belowcaptionskip}{5pt} 
    \centering
    \caption{ADS Construction Cost of Miner.}
    \label{tab:setup_cost}
    \begin{tabular}{cc@{\hspace{5pt}}c@{\hspace{5pt}}c@{\hspace{5pt}}c@{\hspace{5pt}}c@{\hspace{5pt}}c@{\hspace{5pt}}c@{\hspace{5pt}}c}
        \toprule
        \multirow{2}{*}{Dataset} & \multicolumn{2}{c}{NetChain} &  \multicolumn{2}{c}{NetChain$^{+}$}  &  \multicolumn{2}{c}{vChain\textsubscript{1}}  & \multicolumn{2}{c}{vChain\textsubscript{2}}  \\
        \cline{2-9}
         & T & S & T & S & T & S & T & S\\
         \hline
         Wiki & 0.00028 & 5.62 & 0.00057 & 5.68 & 0.00689 & 16.05 & 0.01509 & 25.04\\
         Email & 0.00084 & 8.82 & 0.00198 & 8.93 & 0.02766 &  25.08 & 0.06139 & 42.23\\
         GPlus & 0.00164 & 35.68 & 0.00579 & 35.94 & 0.08649 & 111.30 & 0.14017 & 211.36\\
         \bottomrule

    \end{tabular}
    
    \vspace{0.1em} 
    \par T: ADS construction time (s/block) \quad S: ADS size (KB/block)
\end{table}

\vspace{-1em}
\section{Experimental Evaluation}


\textbf{Experiment Settings.}
In the experiments, we evaluate the performance of NetChain and NetChain$^+$, and compare them with the state-of-the-art on-chain authenticated query framework vChain \cite{conference:vchain}. Note that, we denote vChain\textsubscript{1} and vChain\textsubscript{2} as frameworks with intra-block index only and both intra-block index and inter-block index, respectively.
All considered above frameworks are implemented in about 5k LOCs of C++\footnote{Our code: https://anonymous.4open.science/r/NetChain-4230/}, where we utilize SHA256 in OpenSSL\footnote{\textcolor{black}{The code of OpenSSL: https://openssl-library.org/}} as the cryptographic hash function.
All parties are deployed on a server with Intel(R) Core(TM) i7-10700 CPU@2.60GHz and 80GB RAM, running on Ubuntu 18.04.
Three real-world graph datasets Email, Wiki and GPlus are used as shown in Table.\ref{tab:dataset}. All experiments were repeated 20 times and the average is reported.

\noindent\textbf{Performance Evaluation.} In our experiments, we default to set the size of skiplists (inter-block index) as 5, according to the recommendation of vChain \cite{conference:vchain} due to the
balance between the storage cost and query performance. Then, we set (i) $k=20$ in search query by default,
(ii) for Wiki and Email dataset, the time window size increases from 200 to 1000, 
(iii) for GPlus dataset, the time window size increases from 5K to 25K by default unless otherwise specified.

\noindent{\underline{\emph{ADS Construction Performance.}}} We first evaluate the construction time and storage costs of ADS construction of all the considered frameworks, while the count of objects per block is set as 100, 150, 500 for three datasets, respectively. 
As we can see from Table.\ref{tab:setup_cost}, NetChain has the minimal construction time for every dataset, followed by NetChain$^+$. 
That is because vChain\textsubscript{1} requires expensive accumulator operations in ADS cosntruction and vChain\textsubscript{2} needs to additionally construct skiplists (i.e., intra-block aggregation design), and NetChain$^+$ requires to maintain a global MPT. 
It is worth noting that NetChain and NetChain$^+$ respectively yield up to $53\times$ and $15\times$ improvements in ADS construction time compared with vChain\textsubscript{1}, and yield up to $85\times$ and $31\times$ improvements compared with vChain\textsubscript{2}.
As for ADS storage cost, we observe that NetChain reduces the ADS size by $65\%$ and $77\%$ compared with vChain\textsubscript{1} and vChain\textsubscript{2}, respectively, with the same above reason. 
Note that, the ADS size of NetChain$^{+}$ slightly exceeds that of NetChain since that there is an extra field needs to be added into the block header in NetChain$^+$.
In addition, the size of the block header is 112 bytes for NetChain and vChain\textsubscript{1}, and 144 bytes for NetChain$^{+}$ and vChain\textsubscript{2}.
The above results demonstrate the efficiency of our solutions in ADS construction.

\begin{figure}[t]
\setlength{\abovecaptionskip}{0.cm}
\setcounter{subfigure}{0}
\center
\subfigure[Wiki dataset]
{\includegraphics[height=2cm, width=3.6cm]{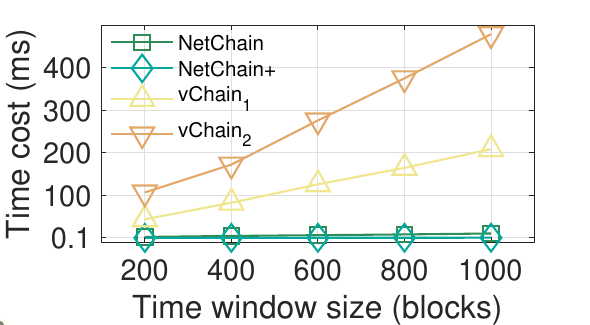}}
\quad
\subfigure[Email dataset]
{\includegraphics[height=2cm, width=3.6cm]{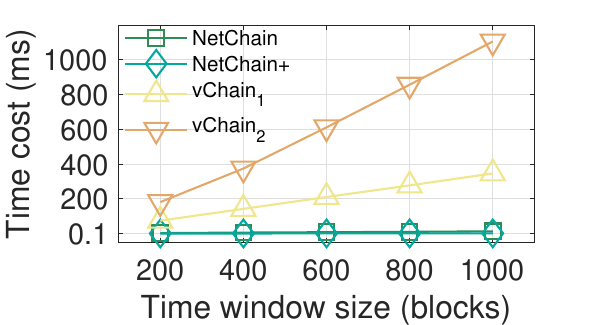}}
\quad
\subfigure[GPlus dataset]
{\includegraphics[height=2cm, width=3.6cm]{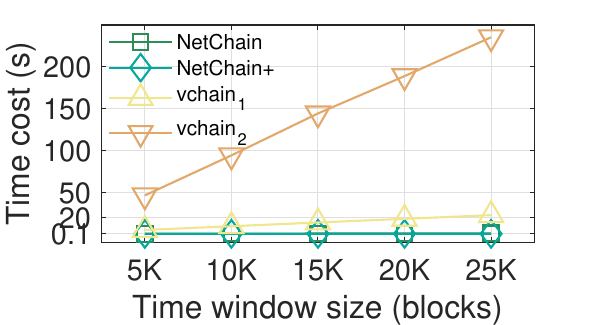}}
\renewcommand{\figurename}{Fig.}
\caption{Search latency.}
\label{exp:search}
\vspace{-2em}
\end{figure}

\begin{figure}[t]
\setlength{\abovecaptionskip}{0.cm}
\setcounter{subfigure}{0}
\center
\subfigure[Wiki dataset]
{\includegraphics[height=2cm, width=3.6cm]{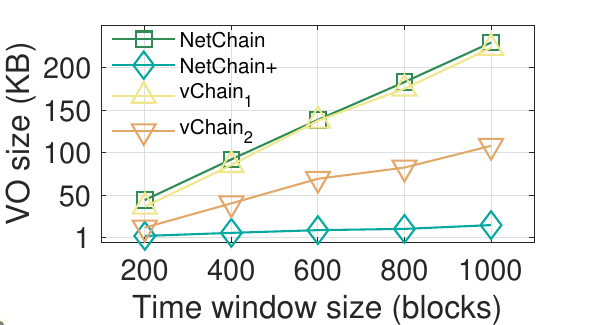}}
\quad
\subfigure[Email dataset]
{\includegraphics[height=2cm, width=3.6cm]{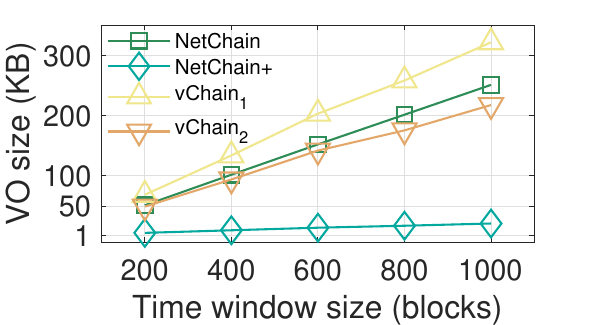}}
\quad
\subfigure[GPlus dataset]
{\includegraphics[height=2cm, width=3.6cm]{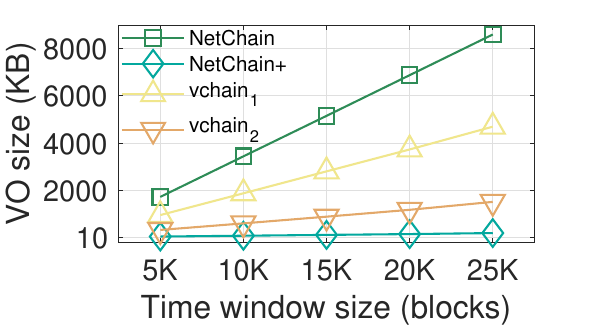}}
\renewcommand{\figurename}{Fig.}
\caption{Response size.}
\label{exp:vo}
\end{figure}

\begin{figure}[t]
\setlength{\abovecaptionskip}{0.cm}
\setcounter{subfigure}{0}
\center
\subfigure[Wiki dataset]
{\includegraphics[height=2cm, width=3.6cm]{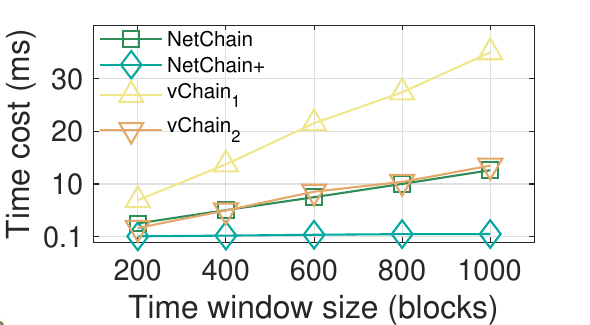}}
\quad
\subfigure[Email dataset]
{\includegraphics[height=2cm, width=3.6cm]{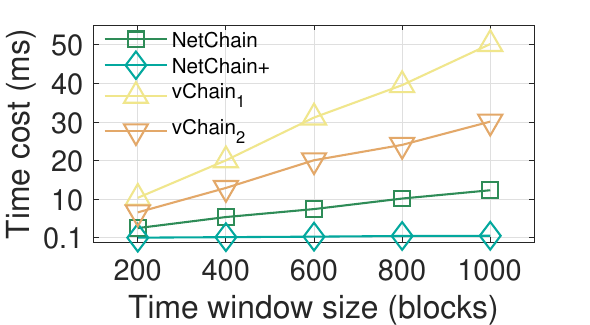}}
\quad
\subfigure[GPlus dataset]
{\includegraphics[height=2cm, width=3.6cm]{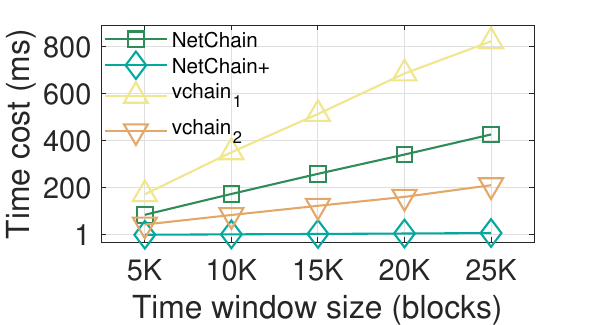}}
\renewcommand{\figurename}{Fig.}
\caption{Verify latency.}
\label{exp:verify}
\vspace{-1.2em}\end{figure}


\noindent{\underline{\emph{Search Performance.}}} Here, we evaluate the search performance of all considered solutions.
Fig.\ref{exp:search} shows that the search time of all solutions increases approximately linearly with time window size. As expected, NetChain$^{+}$ exhibits the best search performance, and the search time of NetChain$^{+}$ is at least $10\times$, $310\times$, and $790\times$ faster than that of NetChain, vChain\textsubscript{1} and vChain\textsubscript{2} respectively. The reason is that NetChain$^{+}$ avoids generating proofs for those non-matched blocks (i.e., the way adopted in NetChain) by using inter-block links, which greatly accelerates the process. Meanwhile, vChain\textsubscript{1} requires to generate proofs for all blocks where lots of extremely expensive accumulator proof operations are invoked, while vChain\textsubscript{2} needs to generate proofs for skiplists additionally. The above results demonstrate the efficiency of our solutions in search performance.

\noindent{\underline{\emph{Communication Cost.}}} Then, we evaluate the response size of all considered solutions. Fig.\ref{exp:vo} reports the the size of response sent from SP.
As expected, NetChain$^{+}$ still yields responses with minimum size, while the response caused by NetChain, vChain\textsubscript{2}, vChain\textsubscript{1} is at least $12.1\times$, $15.3\times$ and $7.6\times$ larger than that of NetChain$^{+}$, respectively.
This is because NetChain$^{+}$ merely adds R and VO of those matched blocks into the response, while other solutions either generate proofs for all blocks in time window (in NetChain and vChain\textsubscript{1}), or needs to generate proofs for invalid skips in addition (in vChain\textsubscript{2}).
We note that NetChain even incurs response $2\times$ larger than that of vChain\textsubscript{1} under GPlus dataset, since a proof in NetChain contains several complete merkle paths, while the one in vChain\textsubscript{1} ideally contains a few of MHT nodes. The above results demonstrate the efficiency of our solutions in communication.

\begin{figure}[t]
\setlength{\abovecaptionskip}{0.cm}
\setcounter{subfigure}{0}
\center
\subfigure[Search latency]
{\includegraphics[height=2cm, width=3.6cm]{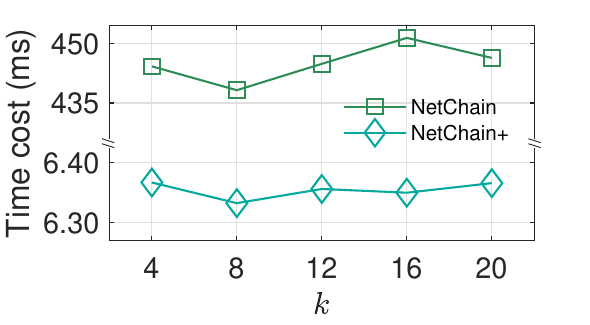}}
\quad
\subfigure[Response size]
{\includegraphics[height=2cm, width=3.6cm]{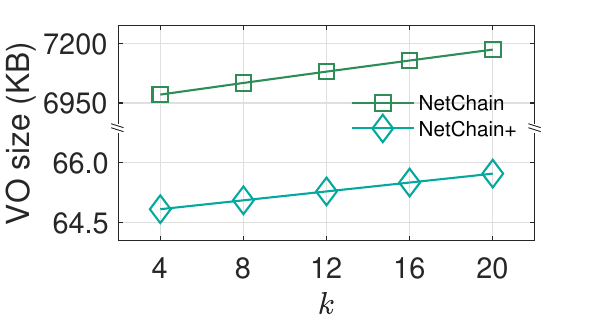}}
\quad
\subfigure[Verify latency]
{\includegraphics[height=2cm, width=3.6cm]{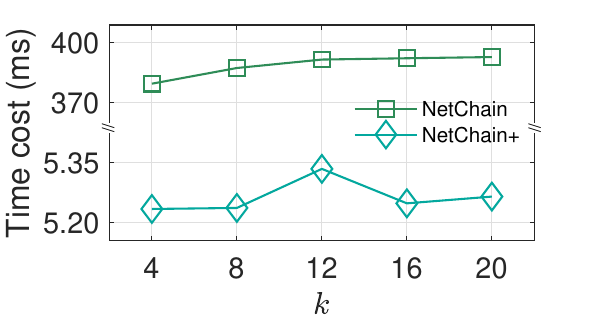}}
\renewcommand{\figurename}{Fig.}
\caption{Performance on different $k$.}
\label{exp:k}
\end{figure}

\noindent{\underline{\emph{Verification Performance.}}} Next, we evaluate the verification time of all considered solutions. As seen from Fig.\ref{exp:verify}, for every dataset, NetChain$^{+}$ performs the best, while vChain\textsubscript{1} performs the worst. More detailed, NetChain$^+$ is at least $22.9\times$, $90.5\times$ and $25.3\times$ faster than NetChain, vChain\textsubscript{1} and vChain\textsubscript{2}, respectively. There are two main reasons, on the one hand, NetChain$^{+}$ only verifies proofs for matched blocks, which mainly calls speedy hash functions.
On the other hand, vChain\textsubscript{1} requires to verify proofs for all matched blocks and frequently invokes expensive accumulator verification operations. In particular, NetChain is faster compared with vChain\textsubscript{2} under small datasets (e.g., Wiki and Email), but slower under large datasets (e.g., GPlus).
The reason is that though having inter-block index, vChain\textsubscript{2} spends huge computation resources on the accumulator proofs, thus vChain\textsubscript{2} only performs better under large datasets benefited from its block pruning strategy. The above results demonstrate the efficiency of our solutions in verification performance.

\noindent{\underline{\emph{Effect of Parameter $k$.}}} Finally, we explore the impact of parameter $k$ on our schemes under GPlus dataset as shown in Fig.\ref{exp:k}, where the time window size is set to 25,000.
The response size and verification time of both NetChain and NetChain$^{+}$ show an overall subtle increasing trend when $k$ grows since the size of R is affected by $k$, but the VO size remains constant, which is further larger than that of R.

\vspace{-1em}
\section{Conclusion}
In this paper, we propose a novel framework NetChain to provide authenticated top-$k$ graph queries over graph-oriented blockchain.
In NetChain, we design a novel authenticated two-layer index as the build-in ADS to reduce ADS construction cost and deal with object (non-)existence efficiently.
To further alleviate the authentication overhead, we propose NetChain$^+$ by customizing an inter-block link construction and a two-round scan mechanism.
Security analysis and extensive experiments show the security and efficiency of our proposed frameworks. In the future, we will strive to support more types of graph queries to provide flexible graph data asset management.

\bibliographystyle{splncs04}
\bibliography{ref.bib}

\end{document}